\documentclass[11pt,letterpaper]{article}

\usepackage{graphicx}
\usepackage{subfigure}
\usepackage{enumerate}
\usepackage{amsmath}
\usepackage[english]{babel}
\usepackage{amssymb}
\usepackage{epsfig}
\usepackage{ifpdf}
\usepackage{color}
\usepackage{float}

\newtheorem{theorem}{Theorem}
\newtheorem{corollary}[theorem]{Corollary}

\newtheorem{lemma}[theorem]{Lemma}
\newtheorem{proposition}[theorem]{Proposition}

\newenvironment{proof}{\par\bproof}{\eproof\(\qed\) \par\medskip}

\newcommand{\qed}{\quad\mbox{\rule{7pt}{7pt}}}

\newcommand{\calC}{\mathcal{C}}

\newcommand{\aNAESAT}{\mathrm{NAE\text{-}3SAT}}
\newcommand{\PNAESAT}{\mathrm{{Positive\text{-}NAE\text{-}3SAT}}}    
\newcommand{\countPNAESAT}{\mathrm{{\# Positive\text{-}NAE\text{-}3SAT}}}  
\newcommand{\SatAssign}{\mathrm{{SatAssign}}}
\newcommand{\CountSatAssign}{\mathrm{{\# SatAssign}}}
\newcommand{\sharpP}{\# \mathrm{P}}

\newcommand{\setof}[1]{\{\,#1\,\}}

\newlength{\first}\newlength{\second}

\newcommand{\smidge}{{\kern .05em}}

\newlength{\subjtolt}

\newcommand{\NP}{{\rm NP}}

\newlength{\saveparindent}
\def\bproof{\begin{rm}\protect\vspace{5pt}\noindent{\bf Proof: }%
\addtolength{\parskip}{4pt}\setlength{\parindent}{0pt}}
\def\eproof{\end{rm}\addtolength{\parskip}{-4pt}%
\setlength{\parindent}{\saveparindent}}

\newcommand{\bprooff}[1]{\begin{rm}\protect\vspace{5pt}%
\noindent{\bf Proof of #1: }\addtolength{\parskip}{4pt}%
\setlength{\parindent}{0pt}}

\begin{document}

\title{Computational Hardness of Enumerating Satisfying Spin-Assignments
   in Triangulations}

\author{Andrea Jim\'enez\thanks{
  Depto.~Ing.~Matem\'{a}tica, U.~Chile.
  Web: \texttt{www.dim.uchile.cl/$\sim$ajimenez}.
  Gratefully acknowledges the support of
    MECESUP UCH0607, and CONICYT via Basal in Applied Mathematics and
    FONDECYT 1090227.}
  \and 
  Marcos Kiwi\thanks{
  Depto.~Ing.~Matem\'{a}tica \&
  Ctr.~Modelamiento Matem\'atico UMI 2807, U.~Chile.
  Web: \texttt{www.dim.uchile.cl/$\sim$mkiwi}.
  Gratefully acknowledges the support of
    CONICYT via Basal in Applied Mathematics and
    FONDECYT 1090227.}
}

\maketitle
\begin{abstract}
Satisfying spin-assignments in triangulations of a surface are states of
  minimum energy of the antiferromagnetic Ising model on triangulations
  which correspond (via geometric duality) to perfect matchings in cubic
  bridgeless graphs.
In this work we show that it is $\NP$-complete to decide
  whether or not a surface 
  triangulation admits a satisfying spin-assignment, and
  that it is $\sharpP$-complete to determine the number of such assignments. 
Both results are derived via an elaborate (and atypical) reduction that maps
  a Boolean formula in 3-conjunctive normal form into a triangulation of an
  orientable closed surface.
\end{abstract}

\section{Introduction}
The Ising model is one of the most studied models 
  in statistical physics. 
Characterizing its behavior on a system~(graph) helps to 
  understand physical phenomena associated to its thermodynamic 
  properties~\cite{citeulike:1289400}. 
The Ising model has been widely studied in lattices 
  and regular structures (see for example~\cite{citeulike:2111390,LV}
  and references therein). 
In contrast, irregular systems have received much less attention,
  probably due to the difficulty of deriving meaningful analytical 
  results.

The number of distinct groundstates of the antiferromagnetic 
  (negative coupling constant) Ising 
  model of a system is called groundstate degeneracy
  and is typically 
  exponentially large as a function of a parameter that measures
  the system's size (the number of nodes of the underlying graph).
The latter translates to 
  nonzero entropy at zero temperature when the system size goes to
  infinity, which in physical terms means that 
  in the thermodynamical limit the spin arrangements of particles in
  the system is disordered. 
This partly explains the considerable attention physicist have
  given to developing techniques for approximating the 
  groundstate degeneracy of a system.

Typically, researchers have focussed on 
  developing techniques for bounding the groundstate degeneracy of a system,
  for example the Transfer Matrix Method~\cite[\S 6.6]{loebl09}.
Instead, informally speaking, in this work we focus on the following two 
  associated computational complexity problems; 
  (1) hardness of deciding whether or not a given system
  admits a satisfying state, and 
  (2) hardness of enumerating groundstates (equivalently,
  computing the groundstate degeneracy) of a given system.
We show that the former problem is $\NP$-complete and the 
  latter is $\sharpP$-complete.

We now make precise the notions discussed above and
  formally state our main results.
  
First, we describe the antiferromagnetic Ising model. 
We say that an 
  embedding of a graph in an orientable closed surface 
  is a \emph{surface triangulation} if each face is bounded by a cycle 
  of length $3$ (in particular there is no loop) --- 
  multiple edges allowed.  
Given a triangulation $T$, let $V(T)$ and $E(T)$ denote
  the node and edge set of $T$.
A mapping 
  $s:V(T)\to\setof{-1,+1}$ will be called a \emph{spin-assignment} 
  (state) to~$T$. 
We refer to $-1$ and $+1$ as \emph{spins}.
The \emph{energy} of a spin-assignment $s$
  of the antiferromagnetic Ising model 
  is defined as $\sum_{uv \in E(T)} \sigma(u) \cdot \sigma(v)$.
A \emph{groundstate} is a spin-assignment of minimum energy.
The number of distinct groundstates that a triangulation $T$	
  admits is often referred to as the \emph{groundstate degeneracy} of $T$.
Clearly, under any spin-assignment to $T$ 
  the ends of at least one edge of each face of a surface triangulation~$T$ 
  are both assigned either $\mbox{-}1$ or $+1$. 
Moreover, a spin-assignment is a groundstate if 
  it has the smallest possible number of edges
  with both its ends being assigned the same spin.
A face $\triangle$ of a surface triangulation~$T$
  is said to be \emph{frustrated} under assignment $s$,
  if $s$ restricted to $V(\triangle)$ is 
  not identically $\mbox{-}1$ or $+1$.
A spin-assignment $s$ to $T$ 
  is said to be \emph{satisfying} 
  (or \emph{feasible}) if every face of $T$ 
  is frustrated under $s$.
Obviously a satisfying spin-assignment is a groundstate.
The converse is true for triangulations that can be embedded
  in the plane~\cite{JKL}.
Nevertheless, the equivalence does not hold in general
  (the reader can verify that 
  the toroidal triangulation depicted in 
  Figure~\ref{fig:triang-non-satisfy}
  does not have satisfying spin-assignments). 
However, note that when satisfying spin-assignments exist, 
  then a groundstate is necessarily a satisfying spin-assignment.

\begin{figure}[h]
\centering
  \ifpdf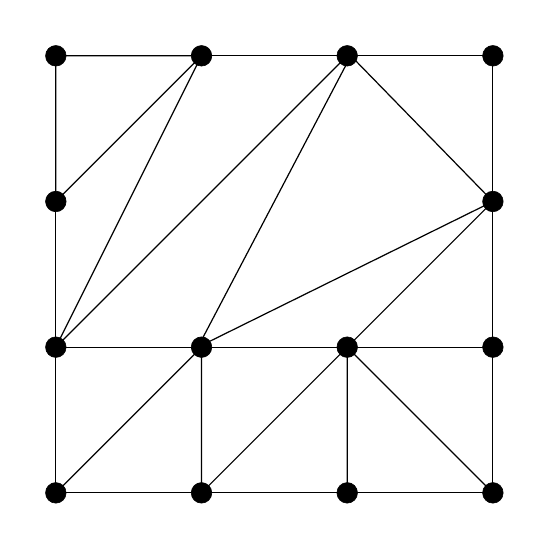\else\input{triang-non-satisfy.pstex_t}\fi
\caption{A triangulation of the torus with no
  satisfying spin-assignment.}
\label{fig:triang-non-satisfy}
\end{figure}

In~\cite{JKL} a relation was established between the 
  groundstate degeneracy of the 
  antiferromagnetic and the number 
  of perfect matchings in cubic bridgeless graphs. 
Specifically, let $T$ be triangulation 
  of a orientable closed surface and $T^*$ its geometric dual.
In~\cite{JKL}, it is shown that the set of edges
  whose ends are assigned the same spin under
  a given satisfying spin-assignment to $T$ correspond 
  to a perfect matching of the cubic bridgeless
  graph $T^*$.  
Moreover, it is shown that
  if $T$ admits a satisfying spin-assignment, then
  the groundstate degeneracy of $T$ 
  is at most twice the number of perfect matchings
  of $T^*$. 
Thus, lower bounds on the groundstate degeneracy of $T$ 
  provide lower bounds on the number distinct of perfect matchings of 
  the cubic bridgeless graph $T^*$.
An old and famous conjecture of Lov\'asz and Plummer, 
  recently positively settled~\cite{EKKKN}, claimed that 
  the number of distinct perfect matchings of $T^*$ is
  exponential in the size (number of nodes) of $T^{*}$.
The relation between number of satisfying spin-assignments 
  of a surface triangulation and the number of perfect matchings
  of bridgeless cubic graphs is another one of our motivations 
  for considering the problem of computing the groundstate
  degeneracy of surface triangulations.

In this work we show that the problem of deciding if a 
  triangulation admits a satisfying spin-assignment
  is $\NP$-complete.  
We also establish that computing the 
  groundstate degeneracy of surface triangulations that 
  admit satisfiable spin-assignments is $\sharpP$-complete. 

%

\subsection{Contributions}\label{subsec:theproblem}
Let $T$ be a surface triangulation. 
For each $v$ in $V(T)$ let $\delta(v)$ denote the set of 
  edges incident to~$v$.
The map $\pi_{v}:\delta(v)\to\delta(v)$ is called 
  \emph{cyclic permutation of the edges incident to $v$} if
  for every $e$ incident to $v$ the edge
  $\pi_{v}(e)$ is the successor of $e$ in the clockwise ordering
  around $v$ defined by the surface embedding of $T$.
The tuple $\pi =(\pi_{v}: v\in V(T))$ is called 
  the \emph{rotation system} of $T$. 
A direct consequence 
  of the Heffter-Edmonds-Ringel rotation principle,
  is that every surface triangulation is uniquely determined, 
  up to homeomorphism,
  by its rotation system~\cite[\S 3.2]{MT}.

Let $\SatAssign$ be the collection of 
  (encodings of) rotation systems of surface triangulations
  that admit a satisfying spin-assignment.
Also, let $\CountSatAssign$ be the function mapping 
  (encodings of) rotation systems of surface triangulations to its 
  number of satisfying spin-assignments.

To see that $\SatAssign$ is in $\NP$, first recall that in order 
  to check that $\pi=(\pi_v:v\in V(T))$ is an instance of $\SatAssign$
  we need not start with a surface.
Indeed, it suffices to check for every $v\in V(T)$ 
  that $\pi_v$ is a cyclic permutation of 
  $\delta(v)=\{uv: \text{$uv$ is in $\pi_u$'s domain}\}$, a 
  task that can be performed in time quadratic in $|V(T)|$ time
  in the Random Access Model.
Then, observe that a certificate of membership in $\SatAssign$
  of a rotation system of a surface triangulation $T$ is simply a 
  spin-assignment $s:V(T)\to\setof{\mbox{-}1,+1}$ and that verifying that
  such an assignment is satisfying amounts to 
  checking that each face $\Delta$ of the surface 
  triangulation $T$ is frustrated under $s$ (which 
  can be checked in $O(1)$ time per face in the Random Access Model).

In this work we establish the following results.
\begin{theorem}\label{th:sat}
$\SatAssign$ is $\NP$-complete.
\end{theorem}

\begin{theorem}\label{th:count}
$\CountSatAssign$ is $\sharpP$-complete.
\end{theorem}
Both of the stated results follow from an elaborate 
  weakly parsimonious reduction~\cite[Definition 2.27]{cks} 
  that maps a Boolean function in $3$-conjunctive normal form 
  to a rotation system of a triangulation (equivalently, to a 
  triangulation embedded on a surface). 
As far as we are aware, this seems to be an atypical reduction,
  whose underlying ideas (e.g.~gadgets) might be of independent
  interest due to their potential usefulness in the study 
  of the computational hardness of other related spin glass
  problems.

\section{Reduction idea and gadgets}
The two main results of this work follow from reductions from
  a well known variant of the standard not-all-equal $3$-satisfiability 
  (abbreviated $\aNAESAT$) problem which is 
  known to be $\NP$-complete even in 
  the absence of negated variables~\cite{TJS}, 
  a variant we denote $\PNAESAT$. 
For completeness sake, we recall  in Figure~\ref{fig:pnaesat}
  the precise definition of $\PNAESAT$. 
Moreover, the counting version of $\PNAESAT$,
  namely $\countPNAESAT$, is $\sharpP$-complete~\cite{CH}. 
See Figure~\ref{fig:cpnaesat} for the precise definition of 
  $\countPNAESAT$.

\begin{figure}[ht]
\begin{tabular}{rp{5.0in}}\hline
\textsc{Problem} & \textsc{$\PNAESAT$} \\ \hline
\textsc{Input} & A Boolean formula $\varphi$ in $3$-conjunctive normal
  form such that each of its clauses $C_1,\ldots,C_m$ 
  has exactly three (all non-negated) literals. \\
\textsc{Output} & 
  \textsc{True} if there is a truth assignment to $\varphi$ such that for each
  clause $C_i$ not all of its variables are assigned the same 
  truth value. \\ \hline
\end{tabular}
\caption{$\PNAESAT$.}\label{fig:pnaesat}
\end{figure}

\begin{figure}[ht]
\begin{tabular}{rp{5.0in}}\hline
\textsc{Problem} & \textsc{$\countPNAESAT$} \\ \hline
\textsc{Input} & A Boolean formula $\varphi$ in $3$-conjunctive normal
  form such that each of its clauses $C_1,\ldots,C_m$ 
  has exactly three (all non-negated) literals. \\
\textsc{Output} & 
  A (binary encoding) of the number of distinct truth value assignments 
  to $\varphi$ such for each
  clause $C_i$ not all of its variables are assigned the same 
  truth value. \\ \hline
\end{tabular}
\caption{$\countPNAESAT$.}\label{fig:cpnaesat}
\end{figure}

The overall strategy we will follow in proving 
  Theorems~\ref{th:sat} and~\ref{th:count}
  is fairly standard, i.e.~we design gadgets where 
  truth values of variables are set (choice gadgets) and 
  gadgets where the truth value of clauses are evaluated 
  (clause gadgets).  
We need to ``carry'' truth values from choice gadgets to clause
  gadgets, and make as many copies of the truth values taken
  by a literal as times they appear in all clauses.
To achieve this task we build so called replicator
  gadgets.
However, the construction of the aforementioned gadgets is 
  quite delicate and non-obvious. 
In general, the main aspects we take care of 
  in the construction of each gadget are existence and uniqueness of 
  satisfying spin-assignments. 
However, there are subtle issues that need to be properly handled
  when building and piecing together the different gadgets.
Below, we describe in separate sections each of the gadgets we 
  will require for the reduction and establish that they
  satisfy certain properties.
First, we introduce some additional terminology and conventions
  we will use throughout the remaining part of this work.

\subsection{Preliminaries} 
Note that given a triangulation $T$, 
  a spin-assignment $s$ to $T$ is satisfying
  if and only if $\mbox{-}s$ is also a satisfying spin-assignment to $T$.
We shall refer to this fact as \emph{duality}.
  We will repeatedly use it in order to reduce the number of cases
  that need to be analyzed in order to establish some of the claims 
  we will make.
  If $T$ admits exactly two satisfying spin-assignments
  $s$ and $\mbox{-}s$, we say that $s$ ($\mbox{-}s$ respectively) is  
  unique up to duality.

A $3$-cycle in a triangulation 
  will be called \emph{positive} for a spin-assignment 
  if at least two of its vertices 
  are assigned spin $+1$. 
Otherwise,
  it will be called \emph{negative}. This concept will be referred
  to as the \emph{sign} of a $3$-cycle.

Henceforth, if $s$ assigns 
  the same spin to all nodes of a subgraph $H$ of $T$
  (respectively all elements of $S\subseteq V(T)$), 
  we say that $H$ (respectively a subset $S$)
  is \emph{monochromatic} under $s$   
Similarly, we say that an edge is monochromatic (respectively 
  non-monochromatic) under $s$ if~$s$ assigns the same 
  (respectively distinct) spins to both ends of the edge.
Monochromatic and non-monochromatic faces are defined analogously
  depending on whether or not its circumscribing cycle is 
  either monochromatic or non-monochromatic.
An edge $e$ in $E(T)$ will be called \emph{serious}
  if and only if
  $e$ is monochromatic under every satisfying spin-assignment to $T$.

The gadgets we build in this work are embedded graphs
  in orientable closed surfaces with some removed disks 
  (with holes) so that each face is bounded by 
  a $3$-cycle and each hole is circumscribed by a $3$-cycle.  
In other words, every gadget may be
  obtained from a triangulation by cutting along
  the boundary of some of its faces (triangles). 
Thence, every term defined 
  for surface triangulations is naturally
  adapted to gadgets so they will be reformulated only 
  in case it is needed.

Throughout this work, serious edges are depicted as thicker lines 
  and surface holes are depicted as gray areas.

\subsection{Choice gadget}
In this section we describe a gadget (a triangulation of a surface)
  that we will associate to Boolean variables in such a way
  that satisfying spin-assignments can be 
  unambiguously interpreted as truth assignments
  to the Boolean variables.
A \emph{choice gadget} is a triangulation 
  as depicted in Figure~\ref{fig:choice_gadget}
  embedded in a toroidal surface with one hole.
The cycle with node set $\{u,v,w\}$ circumscribing 
  the removed triangle of the choice gadget $L$,
  henceforth denoted by~$\calC_L$, will be referred to as the \emph{variable 
  cycle} of~$L$ and the edge $uw$ will be called the 
  \emph{fundamental edge} of~$L$ (see Figure~\ref{fig:choice_gadget}).

Our reduction will associate to each variable $x_i$ 
  a choice gadget $L_i$. 
A satisfying spin-assignment will be interpreted as setting
  $x_i$ to \textsc{True} if $\calC_{L_i}$ ends up being monochromatic,
  and \textsc{False} otherwise.
The key functionality that we will show a choice
  gadget provides is that
  it has a unique up to duality satisfying spin-assignment 
  where the variable cycle is monochromatic 
  (respectively, non-monochromatic). 
Furthermore, choice gadgets will also be used as auxiliary building
  blocks in the construction of another type of gadget we will soon
  encounter.

\begin{figure}[h]
\centering
\ifpdf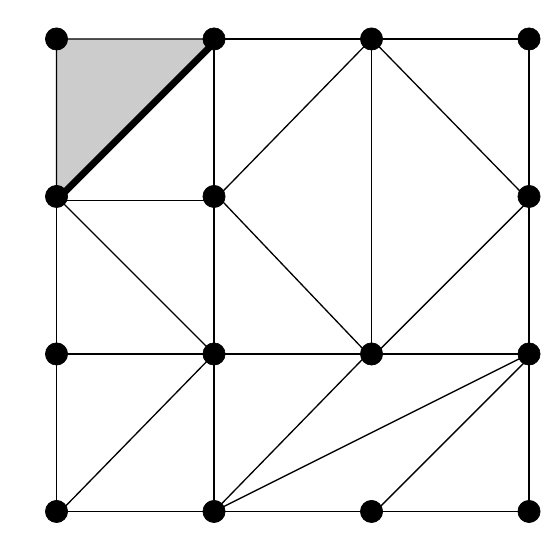\else\input{choice-gadget.pstex_t}\fi
\caption{Choice gadget (region in gray depicts a surface hole).}
\label{fig:choice_gadget}
\end{figure}

The following result encapsulates 
  the most relevant properties of choice gadgets.
\begin{proposition}\label{prop:choice_properties}
Let $L$ be a choice gadget.
The fundamental edge of the 
  choice gadget is serious.
Moreover, there exists a unique up to duality 
  feasible spin-assignment to $L$ where the variable cycle 
  $\calC_{L}$ of $L$ is monochromatic 
  (respectively, non-monochromatic).
\end{proposition}
\begin{proof} 
To prove the first statement, by duality, 
  it suffices to show that there is no feasible spin-assignment
  extension to $L$ when 
  node $v$ (node labels as in Figure~\ref{fig:choice_gadget}) 
  is assigned spin $+1$ and 
  the fundamental edge $uw$ is assigned spins $+-$ or $-+$.
In Figure~\ref{serious_edge_choice_1} we work out the case
  where $uw$ is assigned spins $+-$; a subindex $i$  
  accompanying a $+$ or $-$ sign indicates that the spin is
  forced by the spin-assignments with smaller indices in order
  for the assignment to be satisfiable --- if spins assigned
  to the vertices of a triangle are forced to be all of the same
  sign, then no satisfying assignment can exist under the given
  initial conditions. 
The case when $uw$ is assigned spins $-+$ 
  is dealt with in the same way and worked out in 
  Figure~\ref{serious_edge_choice_2}.
This establishes that the fundamental edge of $L$ is serious.

\begin{figure}[h]
\centering
\subfigure[Assignment forced by fixing $uw$ to $+-$.
  Forced monochromatic triangular faces are labeled by 
  $\rightarrow\leftarrow$.]
{
\ifpdf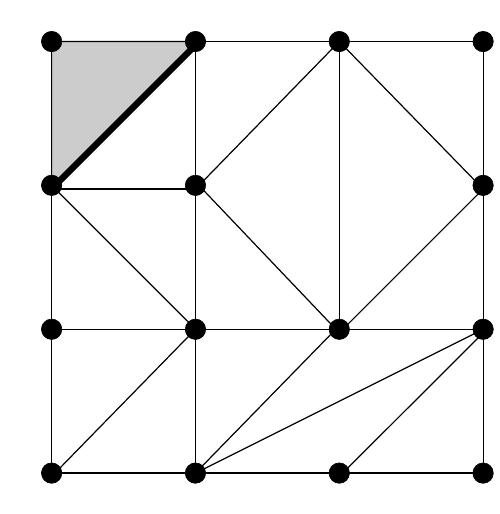\else\input{choice_serious_1.pstex_t}\fi
  \label{serious_edge_choice_1}
}
\hspace{2em}
\subfigure[Assignment forced by fixing $uw$ to $-+$.
  Forced monochromatic triangular faces are labeled by 
  $\rightarrow\leftarrow$.]
  {\ifpdf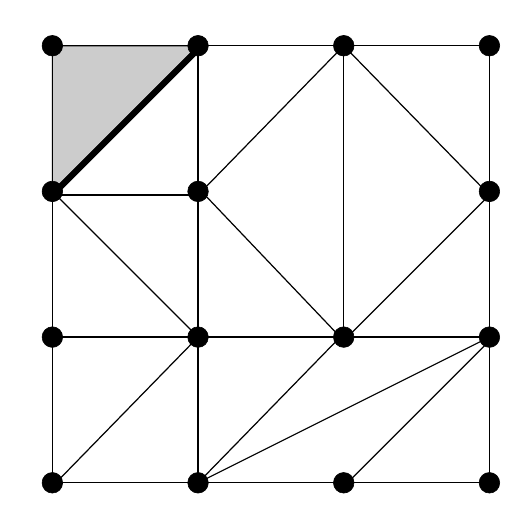\else\input{choice_serious_2.pstex_t}\fi
  \label{serious_edge_choice_2}
  }
\caption{\mbox{}}
\end{figure}

We now establish the claimed existence and uniqueness. 
Since the fundamental edge of $L$ is serious,
  if $s$ is a satisfying spin-assignment 
  to $L$, then $L$'s fundamental edge is monochromatic
  under $s$.
Therefore, again by duality, it is enough to prove that in the 
  following two cases there exist exactly 
  one feasible spin-assignment extension:
  (a) when the variable cycle 
  $\calC_{L}$ (i.e.~$uwv$) of~$L$ 
  is assigned spin $+++$ (the monochromatic case),  
  and (b) when it is assigned $++-$ (the non-monochromatic case). 
In Figure~\ref{fig:choice_unique_mono}, the unique 
  satisfying spin-assignment to $L$ when
  nodes of its variable cycle are assigned $+1$ is exhibited.

\begin{figure}[h]
\centering
\subfigure[Unique 
  satisfying spin-assignment forced by the spin-assignment 
  $+++$ to the variable cycle of $L$.]{
  \ifpdf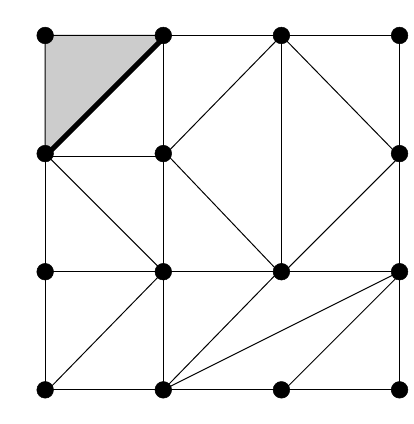\else\input{choice_unique_mono.pstex_t}\fi
  \label{fig:choice_unique_mono}}
\hspace{2em}
\subfigure[Unique satisfying spin-assignment to~$L$,
  up to duality, when $L$'s variable cycle is non-monochromatic.]{
  \ifpdf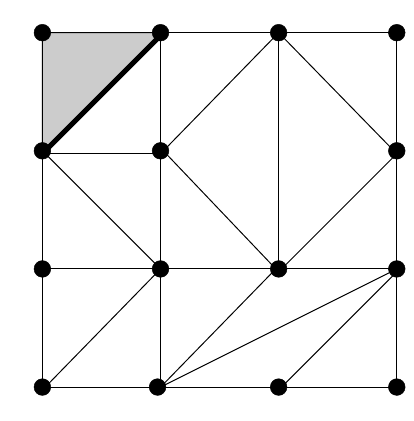\else\input{choice_unique_nonmono_1.pstex_t}\fi
  \label{fig:choice_unique_nonmono_1}
}
\hspace{2em}
\subfigure[Spin assignments forced by the assignment of spin~$-1$
  to $u'$ (forced monochromatic faces are labeled by 
  $\rightarrow\leftarrow$).]{
  \ifpdf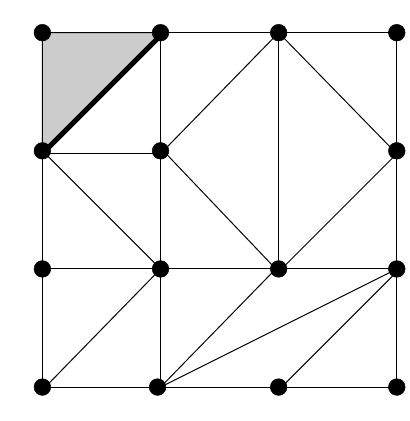\else\input{choice_unique_nonmono_2.pstex_t}\fi
  \label{fig:choice_unique_nonmono_2}
}
\caption{}
\end{figure}
For the non-monochromatic case, by duality and since
  $L$'s fundamental edge is serious, it suffices to consider
  the situation where $L$'s fundamental edge is assigned
  spins $++$.
Two subcases arise, depending on whether or not the spin $+1$
  is assigned to node $u'$ (node labels as in Figure~\ref{fig:choice_gadget}) 
  --- each subcase is 
  worked out separately in Figures~\ref{fig:choice_unique_nonmono_1} 
  and~\ref{fig:choice_unique_nonmono_2}.
\end{proof}

\subsection{Replicator gadget}
A variable may appear several times
  in different clauses (or even multiple times in the same clause)
  of a Boolean formula in three conjunctive normal form.
Given that the truth value a variable, say $x_i$,
  will be unambiguously set by the values taken by
  a satisfying spin-assignment at the associated choice gadget, say $L_i$,
  we require a way of ``replicating'' the encoding of 
  the truth value of $x_i$ as many times as 
  $x_i$ appears in the collection of formula clauses.
To achieve this goal, to every choice gadget we will connect
  a special gadget, namely a \emph{$k$-replicator gadget}. 
When the value $k$ is clear from context or is not relevant,  
  we will just write replicator gadget instead of $k$-replicator gadget.
For each $k>0$, a $k$-replicator gadget will be a triangulation
  embedded in a surface of large genus (depending on $k$) 
  with $2^k +1$ holes.
The purpose of a $k$-replicator gadget is to
  generate $2^k$ copies of the truth value encoded by a satisfying
  spin-assignment to the choice gadget to which the replicator gadget is 
  connected.

To construct a $k$-replicator gadget 
  we will glue together (in a particular way soon to be discussed) 
  $2^{k}-1$ so called block-replicator gadgets.
A \emph{block-replicator gadget} is a triangulation~$R$ of an
  orientable closed surface of genus $4$ with three holes.
A block-replicator gadget is built by gluing together three choice gadgets
  and the surface triangulation of the torus with six removed triangles
  depicted in Figure~\ref{fig:replicator-gadget}.\footnote{Here, choice gadgets
  are used as auxiliary gadgets. This auxiliary gadgets will not be 
  associated to Boolean variables. The reason why we rely on this auxiliary
  choice gadgets is solely because of one of the properties we have shown 
  they exhibit. Specifically, the fact that fundamental edges of choice
  gadgets are serious.}
Specifically, the construction 
  takes three choice gadgets, say $\hat{L}$, $\bar{L}$ and 
  $\tilde{L}$ as depicted in Figure~\ref{fig:choice_gadget}, and
  identifies the variable cycle $\calC_{\hat{L}}$
  (respectively, $\calC_{\bar{L}}$ and $\calC_{\tilde{L}}$ ) 
  of the choice gadget $\hat{L}$ 
  (respectively, $\bar{L}$ and $\tilde{L}$)
  with the cycle $\hat{\calC}=u'xu$ (respectively,  
  $\bar{\calC}=xyw$ and $\tilde{\calC}=zw'v$) of the 
  block-replicator gadget $R$ as depicted in Figure~\ref{fig:replicator-gadget}.
The identification is done in such a way 
  that edge $u'x$ (respectively, $xy$ and $zw'$) 
  of the cycle $\hat{\calC}$ (respectively,
  $\bar{\calC}$ and $\tilde{\calC}$) coincides with the 
  fundamental edge of the choice gadget $\hat{L}$ 
  (respectively, $\bar{L}$ and  $\tilde{L}$).
\begin{figure}[h]
\centering
\ifpdf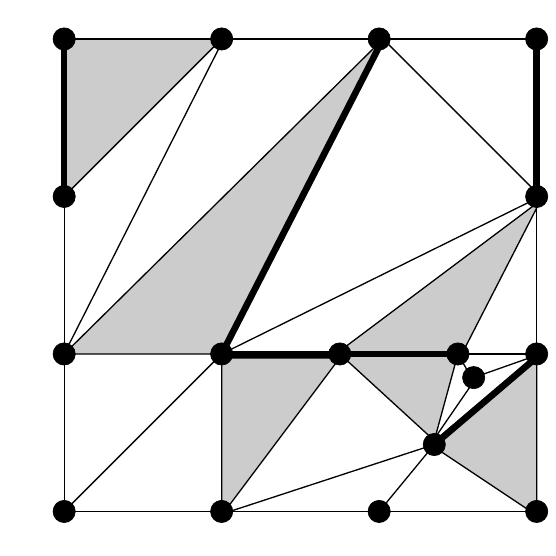\else\input{replicator_sharpP.pstex_t}\fi
\caption{Block-replicator gadget.}\label{fig:replicator-gadget}
\end{figure}
Clearly, under this construction, each block-replicator
  gadget has exactly three holes, because after gluing 
  the choice gadgets to the surface triangulation depicted in
  Figure~\ref{fig:replicator-gadget}, 
  the holes circumscribed
  by the cycles $\hat{\calC}$, $\bar{\calC}$ and
  $\tilde{\calC}$ disappear. 
The length $3$ cycle $uvw$ circumscribing one of $R$'s hole
  is referred to as the 
  \emph{incoming cycle} (node labels as in Figure~\ref{fig:replicator-gadget}).
The length $3$ cycles circumscribing the other two holes 
  of $R$ will be called the
  \emph{outgoing cycles}.
Moreover, edges $vw$, $v'u'$ and $xy$ will be referred to as 
  \emph{fundamental edges} of
  $R$ (see again Figure~\ref{fig:replicator-gadget} for node labeling scheme).

The attentive reader might wonder whether the described block-replicator
  gadget is indeed a surface triangulation. 
Specifically, whether indeed every ``surface'' point has an open neighborhood 
  homeomorphic to some open subset of the Euclidean plane.
This is indeed the case.
Moreover, a block-replicator gadget has 
  the following key property, henceforth referred to as 
  \emph{intersection property}: 
  the incoming and outgoing cycles of a block-replicator gadget 
  do not share vertices. 
The intersection property implies that the surface on which the 
  block-replicator gadget is embedded can be smoothly deformed into
  the one depicted in Figure~\ref{fig:block_replicator_diagram}.

\begin{figure}[h]
\centering
\ifpdf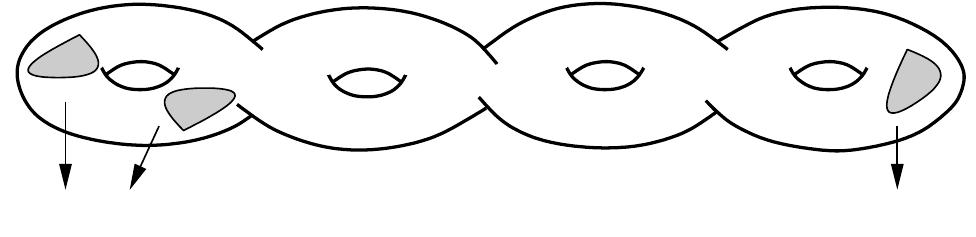\else\input{block_replicator_diagram.pstex_t}\fi
\caption{Block-replicator gadget sketch.}\label{fig:block_replicator_diagram}
\end{figure}
  
The key purpose of each block-replicator gadget is to enforce 
  that if the incoming cycle is monochromatic 
  (respectively, non-monochromatic), 
  then both outgoing cycles will be monochromatic 
  (respectively, non-monochromatic). 
Moreover, we will see that the 
  block-replicator gadget inverts the sign of the incoming
  cycle, namely if the incoming cycle is positive 
  (respectively, negative) the outgoing cycles are
  negative (respectively, positive). Actually, 
  much more is true. 
Formally, we have the following 
  results concerning block-replicator gadgets.
  
\begin{proposition}\label{prop:serious_edges_replicator}
Let $R$ be a block replicator gadget. 
  Fundamental edges of $R$ are 
  serious. 
In particular,   in every satisfying spin-assignment to $R$ both 
  outgoing cycles have the same sign
  and opposite to the sign of the incoming cycle.
\end{proposition}

\begin{proof} 
Given that $xy$ is the fundamental edge 
  of the variable cycle $\calC_{\bar{L}}=xyw$ of the choice gadget $\bar{L}$,
  by Proposition~\ref{prop:choice_properties}, we
  have that $xy$ is serious.
To prove that $vw$ is serious, by duality, it suffices
  to show that when $vw$ is assigned
  $+-$ and $u$ is assigned $+$ and $-$, there is no 
  feasible spin-assignment extension to $R$. 
These two situations are worked out in Figure~\ref{replicator_serious_1}.

\begin{figure}[h]
\centering
\subfigure[\mbox{}]{
  \ifpdf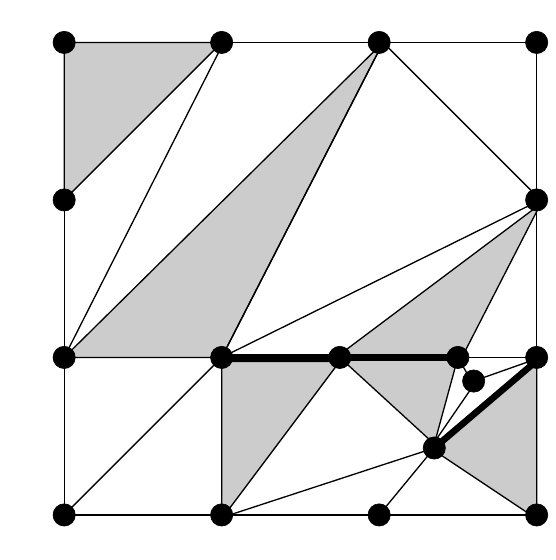\else\input{replicator_serious_1.pstex_t}\fi
}
\hspace{2em}
\subfigure[\mbox{}]{
  \ifpdf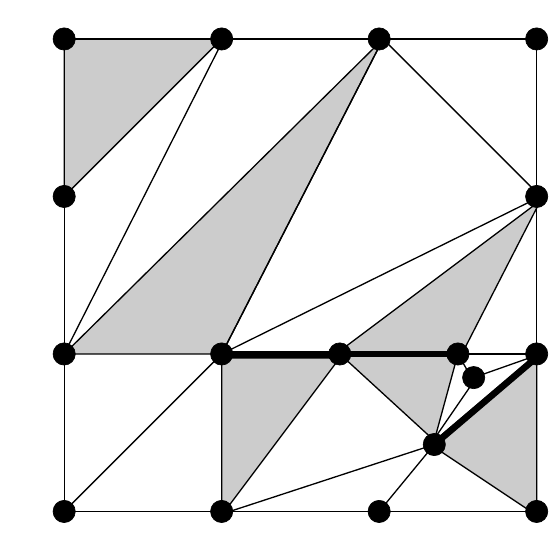\else\input{replicator_serious_2.pstex_t}\fi
} 
\caption{In (a),  
  satisfying spin-assignment forced by fixing
  the outgoing cycle $uvw'$ to $-++$ and in (b) to $-+-$. 
  Serious edges are shown as thick lines.
    Note that fixing the spin of an end of a serious edge immediately
    forces the spin of its other end.}
\label{replicator_serious_1} 
\end{figure}

\begin{figure}[h]
\centering
\subfigure[\mbox{}]{
  \ifpdf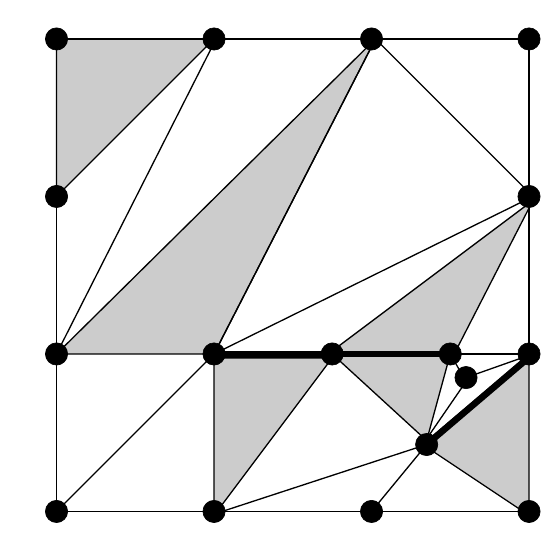\else\input{replicator_serious_3.pstex_t}\fi
}
\hspace{2em}
\subfigure[\mbox{}]{
  \ifpdf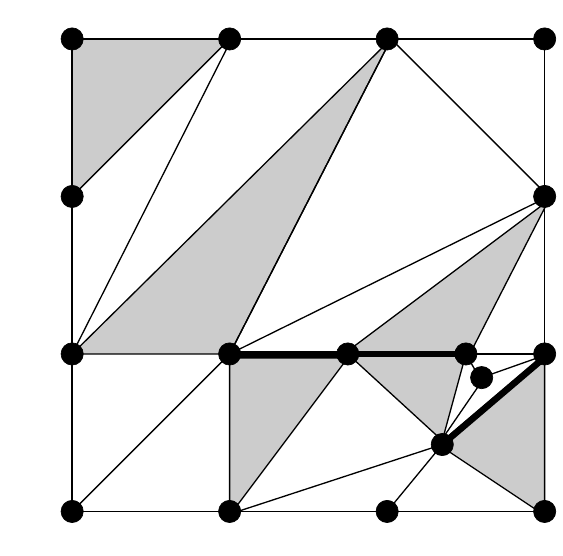\else\input{replicator_serious_4.pstex_t}\fi
} 
\caption{In (a),  
  satisfying spin-assignment forced by fixing
  the outgoing cycle $u'v'w'$ to $-++$ and in (b) to $-+-$. 
  Serious edges are shown as thick lines.
    Note that fixing the spin of an end of a serious edge immediately
    forces the spin of its other end.}
\label{replicator_serious_2} 
\end{figure}

We proceed as above to prove now that $u'v'$ is serious. 
In Figure~\ref{replicator_serious_2}, we show that in the cases
  where $u'v'$ is assigned $-+$ and $w'$ is assigned $+$ and $-$ there is
  no feasible spin-assignment extension to $R$.
\end{proof}

\begin{proposition}\label{prop:replicator}
Let $R$ be a block-replicator gadget. 
  Then, there exists a unique up to duality satisfying 
  spin-assignments to $R$ for which
  the incoming cycle is monochromatic (respectively, non-monochromatic).
Moreover, if $s$ a satisfying spin-assignment
  to $R$, one of the two following statements hold:
\begin{enumerate}
 \item the incoming and outgoing cycles are all monochromatic
  with the incoming cycle positive 
  (respectively, negative) and both outgoing
  cycles negative (respectively, positive), or
 \item the incoming and outgoing cycles are all non-monochromatic
 with the incoming cycle positive 
  (respectively, negative) and both outgoing
  cycles negative (respectively, positive).
\end{enumerate}
\end{proposition}
\begin{proof} 
To prove existence and uniqueness of the 
  satisfying spin-assignment when the 
  incoming cycle is monochromatic, 
  by duality, it is enough to prove that if 
  all nodes in the incoming cycle 
  $uvw$ are assigned $+1$, then there exists a 
  unique feasible spin-assignment extension to $R$  
  (node labels as in Figure~\ref{fig:replicator-gadget}). 
This situation is
  worked out in Figure~\ref{fig:replicator-mono}.
On the other hand, by Proposition~\ref{prop:serious_edges_replicator}, if $s$ is
  a satisfying spin-assignment to $R$, the edge $vw$ belonging to the
  incoming cycle is monochromatic (because it is serious). 
Thus, by duality, we can assume that $s$ assigns 
  to the incoming cycle $uvw$  spins $+++$ or $-++$. 
Therefore, to establish
  existence and uniqueness of the satisfying spin-assignment
  when the incoming cycle is non-monochromatic, 
  it will suffice to 
  show that there exists unique satisfying spin-assignment extension
  to $R$ when the incoming cycle $uvw$ is assigned spin $-++$.
This case is studied in Figure~\ref{fig:replicator-non-mono}.
We need to check that each of the spin-assignments depicted
  in Figures~\ref{fig:replicator-mono} and~\ref{fig:replicator-non-mono}
  have a unique extension to the block-replicator gadget, even when 
  the auxiliary choice gadgets are glued to the block-replicator
  gadget via proper identification of 
  $\hat{\calC}, \bar{\calC}, \tilde{\calC}$ 
  (see labels in Figure~\ref{fig:block_replicator_diagram})
  and the variable cycles of the auxiliary choice gadgets.
Proposition~\ref{prop:choice_properties} and the fact that 
  the spin-assignments depicted in Figures~\ref{fig:replicator-mono} 
  and~\ref{fig:replicator-non-mono} completely determine the spins
  of the nodes of   $\hat{\calC}, \bar{\calC}, \tilde{\calC}$ imply
  that the spin-assignment extensions to the whole 
  block-replicator gadget are indeed feasible and unique.

\begin{figure}[h]
\centering
\subfigure[Case where the incoming cycle is positive
           and monochromatic. 
	   Note that the outgoing cycles are forced
	   to be negative and monochromatic.]{
  \ifpdf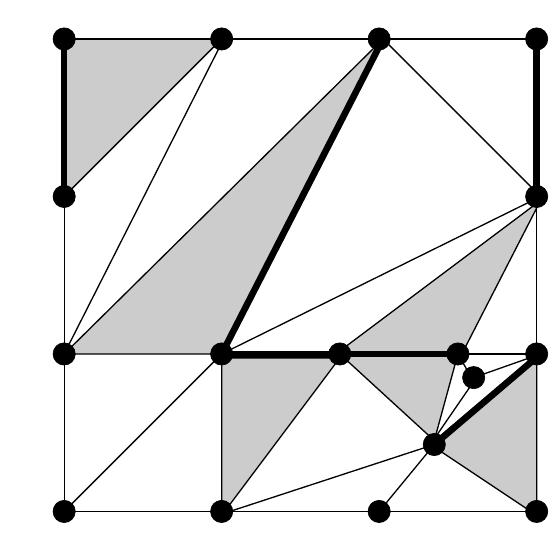\else\input{repl_mono_mono.pstex_t}\fi
  \label{fig:replicator-mono}}
\hspace{2em}
\subfigure[Case where the incoming cycle is positive
           and non-monochromatic. 
	   Note that outgoing cycles are forced
	   to be negative and non-monochromatic.]{
  \ifpdf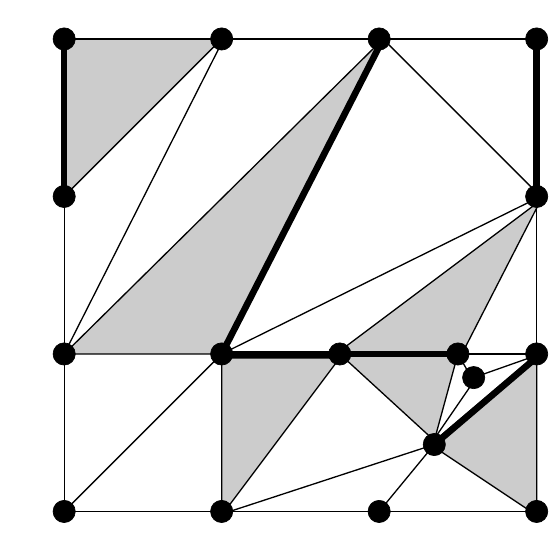\else\input{repl_nomono_nonmono.pstex_t}\fi
  \label{fig:replicator-non-mono}
} 
\caption{\mbox{}} 
\end{figure}

The remaining part of the claimed result can be ascertained
  by inspecting in Figures~\ref{fig:replicator-mono} 
  and~\ref{fig:replicator-non-mono}
  the satisfying spin-assignments
  forced by the (non) monochromaticity of the incoming 
  cycles.
\end{proof}

We are now ready to describe the construction of 
  a $k$-replicator gadget.
Take $2^k - 1$ block-replicator gadgets $R_1, R_2, R_3, \ldots, R_{2^k -1}$.
Identify the outgoing cycles of $R_1$
  with the incoming cycles of 
  the block replicator gadgets $R_2$
  and $R_3$ so that the fundamental edges of $R_1$ 
  that belong to the outgoing cycles and 
  the fundamental edges of $R_2$ and $R_3$ that 
  belong to the incoming cycles coincide.
Continue in this way piecing together new block 
  replicator gadgets and identifying fundamental edges, and
  construct a ``rooted binary tree of depth $k$'' type structure  
  of block-replicator gadgets. 
Let $R_1$ denote the block-replicator gadget at the 
  ``root'' of the tree, and let 
  $R_{2^{k-1}}, R_{2^{k-1} +1}, \ldots, R_{2^{k} -1}$
  denote the block-replicator gadgets at the ``leaves'' of the tree.
The incoming cycle of $R_1$ will be referred to as the
  \emph{starting cycle} of the $k$-replicator gadget $R^k$
  and the outgoing cycles of the block replicator gadgets
  $R_{2^{k-1}}, R_{2^{k-1} +1}, \ldots, R_{2^{k} -1}$
  will be called \emph{end cycles} of $R^k$. 
Moreover, fundamental edges of the block-replicator gadget
  belonging to the incoming cycle of $R_1$ and
  to the outgoing cycles of 
  $R_{2^{k-1}}, R_{2^{k-1} +1}, \ldots, R_{2^{k} -1}$
  will be referred to as fundamental edges of
  the $k$-replicator gadget.

Note that each $k$-replicator
  gadget is a triangulation of an orientable closed
  surface of genus $4 \cdot (2^{k}-1)$ with $2^k +1$ holes.
Furthermore, the intersection property
  of the block-replicator gadgets is trivially transferred
  to $k$-replicator gadgets; namely, the starting and
  end cycles of a $k$-replicator gadget do not share vertices.

In our reduction, the starting cycle of each $k$-replicator gadget $R^k$
  will be identified with the variable cycle of a choice gadget, say $L$. 
By Proposition~\ref{prop:replicator}, this guarantees that
  the end cycles of $R^k$ will be monochromatic if and only if 
  the variable cycle of $L$ is monochromatic.
It is somewhat unfortunate 
  that the block-replicator gadgets generate, at its
  outgoing cycles, encodings of opposite signs as the one 
  of its incoming cycle.  
By taking $k$ even, we can guarantee that 
  each of the end cycles of $R^k$ will have the same 
  chromaticity (monochromatic or nonmonochromatic) 
  and sign as the variable cycle of the choice
  gadget $L$.
The following result captures all relevant properties 
  we will need that are satisfied
  by replicator gadgets.
The reader can easily check that the claimed properties 
  are immediately inherited from those satisfied by 
  block-replicator gadgets.
\begin{corollary}\label{coroll:replicatorgadget}
Let $k$ be a positive integer and let $R^k$ be a $k$-replicator gadget. 
The following statements hold:
\begin{enumerate}
\item \label{st1_rep} Fundamental edges of $R^k$ are serious.
\item \label{st2_rep} For any 
  satisfying spin-assignment to $R^k$, the starting cycle
  and the end cycles have the same sign.
\item \label{st3_rep} For any 
  satisfying spin-assignment to $R^k$, the starting cycle
  and the end cycles are all either monochromatic or non-monochromatic.
\item \label{st4_rep} There is a unique up duality satisfying
  spin-assignment to $R^k$ so that the starting cycle
  and the end cycles are all monochromatic (respectively, non-monochromatic).
\end{enumerate}
\end{corollary}

As we have already mentioned, the starting cycle of a replicator gadget, 
  say $R$, 
  will be identified with a variable cycle of a choice gadget, say $L$. 
Assuming that $L$ is in turn associated to a formula variable, say $x$, 
  it follows that in any satisfying spin-assignment 
  all end cycles of~$R$ encode the same truth value of $x$
  encoded by the variable cycle of $L$.
Eventually, some end cycles of $R$
  will be identified with cycles of the (next to be described) 
  clause gadgets associated to formula clauses where $x$ appears.
If the total number of appearances of $x$ in formula clauses 
  is $t$, then $R$ will be a $k$-replicator gadget where  
  $k$ is the smallest positive even integer greater or equal than $\log_{2}t$.
Thus, after identifying end cycles of $R$ with cycles in clause
  gadgets, we might end up with non-identified end cycles (a
  situation that occurs whenever $\log_{2}t$ is not a positive even integer).
The holes circumscribed by such end cycles
  will need to be ``capped'' in order so at the end of our 
  reduction we do indeed generate a surface triangulation.
Moreover, holes will need to be ``capped'' in such a way that 
  the properties satisfied by replicator gadgets are preserved.
To achieve this goal, when necessary, 
  we will identify an end cycle   
  with the outer cycle of a 
  \emph{cap gadget} as depicted in Figure~\ref{fig:cap}.

\begin{figure}[h]
\centering
\ifpdf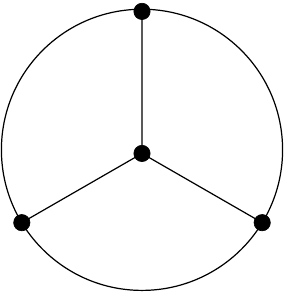\else\input{cap.pstex_t}\fi
\caption{A cap.}\label{fig:cap}
\end{figure}

The following statement is trivial.
\begin{proposition}\label{prop:cap}
For any spin-assignment to the outer cycle of a cap gadget, there exists
  a unique satisfying spin-assignment extension to the whole cap.
\end{proposition}

\subsection{Clause gadget}
A \emph{clause gadget} is a
  toroidal triangulation with three 
  holes as depicted in Figure~\ref{fig:clause-gadget}.
The cycles circumscribing the holes of the 
  clause gadget will be called \emph{literal cycles}. 
Moreover, edges
  $uw$, $v'w$ and $v'u$ will be referred to as 
  \emph{fundamental edges} of the clause gadget 
  (depicted as thicker lines in Figure~\ref{fig:clause-gadget}).

\begin{figure}[h]
\centering
\ifpdf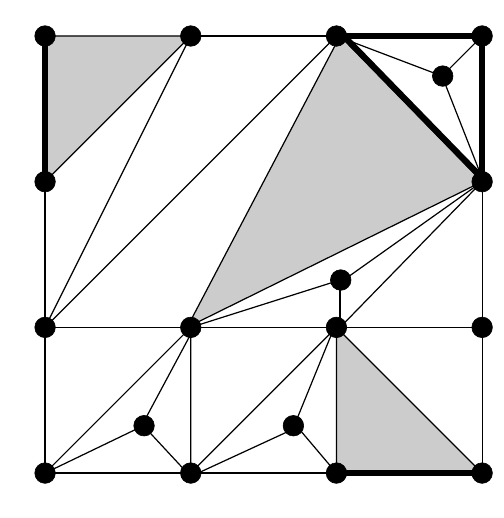\else\input{clause-gadget_sharpP.pstex_t}\fi
\caption{Clause gadget.}\label{fig:clause-gadget}
\end{figure}
  
As already mentioned in the preceding section, 
  we will eventually identify end cycles of replicator 
  gadgets with literal cycles in such a way that
  fundamental edges of the replicator and
  clause gadgets coincide.
In our reduction, replicator gadgets will ``carry'' 
  from choice gadgets towards clause gadgets 
  the encodings of the truth values of formula variables. 
The clause gadget is built in such a way as to 
  allow a unique up to duality satisfying spin-assignment extension
  if and only if not all the truth value encodings
  ``arriving'' to the clause gadget represent the same truth value.

Unfortunately, fundamental edges of clause gadgets
  are not serious. 
However, once every literal cycle 
  of a clause gadget is identified with 
  an end cycle of a replicator
  gadget, fundamental edges of the clause
  gadget will become serious in the triangulation 
  thus formed --- 
  since fundamental edges of the replicator gadgets are serious,
  and because fundamental
  edges of the clause and replicator gadgets will
  be identified.
This explains why when stating the following claims
  we assume seriousness of fundamental edges
  of the clause gadgets.  
The functionality provided by a clause gadget is summarized by the 
  next results, the first of which is obvious.

\begin{proposition}\label{prop:serious_edges_clause}
Let $C$ be a clause gadget. Assume that fundamental 
edges of $C$ are serious. Then, 
for any satisfying spin-assignment,
all literal cycles of $C$ have the same sign.
\end{proposition}

\begin{proposition}\label{prop:clause-nonfeasible}
Let $C$ be a clause gadget. Assume that fundamental 
edges of $C$ are serious.
In the following cases there is no 
  satisfying spin-assignment extension to $C$:
\begin{enumerate}
\item \label{item:clause0} 
when all literal cycles of $C$ are monochromatic, and
\item \label{item:clause1} 
when all literal cycles of $C$ are non-monochromatic.
\end{enumerate}
\end{proposition}
\begin{proof} 
To prove the first claim,
  by duality and Proposition~\ref{prop:serious_edges_clause},
  it suffices to show that if all nodes in each literal
  cycle are assigned $+1$, then there is no feasible spin-assignment
  extension to $C$. 
This case is worked out in Figure~\ref{fig:all_equal_literals_a}.
\begin{figure}[ht]
\centering
\subfigure[Case where all literal cycles are monochromatic.]{
\ifpdf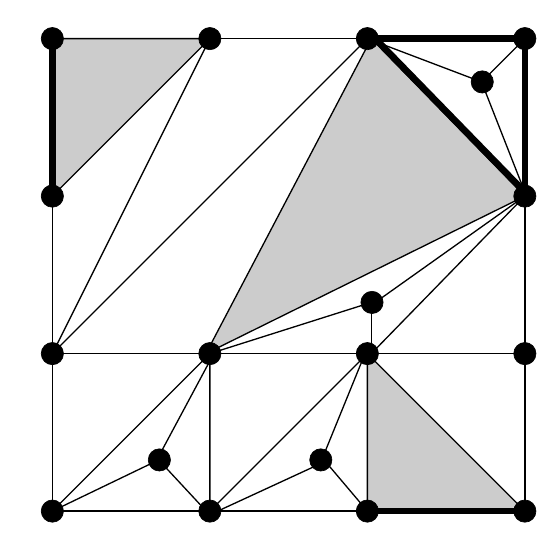\else\input{clause_gadget_all_mono.pstex_t}\fi
\label{fig:all_equal_literals_a}
}
\hspace{2em}
\subfigure[Case where all literal cycles are non-monochromatic.]{
\ifpdf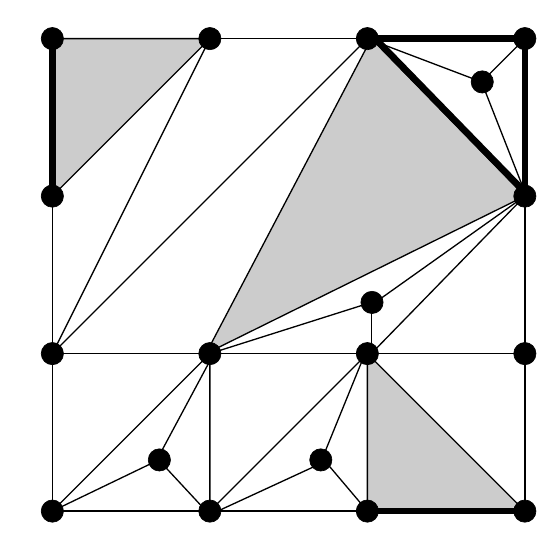\else\input{clause_gadget_all_nonmono.pstex_t}\fi
\label{fig:all_equal_literals_b}
}
\caption{\mbox{}}
\label{fig:all_equal_literals}
\end{figure}

To establish the second claim, it suffices to show that the 
  same conclusion holds when 
  the literal cycles $uwv$,
  $v'wu'$ and $v'uw'$ are all assigned $++-$ 
  (node labels as in Figure~\ref{fig:clause-gadget}).
This case is worked out in Figure~\ref{fig:all_equal_literals_b}.
\end{proof}

\begin{proposition}\label{prop:clause-feasible}
Let $C$ be a clause gadget. 
Assume that the fundamental edges of $C$ are serious. 
If exactly one literal cycle of $C$ is monochromatic 
  (respectively, non-monochromatic)
  there is a unique up to duality satisfying spin-assignment extension to $C$.
\end{proposition}
\begin{proof} 
By duality and Proposition~\ref{prop:serious_edges_clause}
  the monochromatic case holds if we show that in each of the
  following situations there is exactly one feasible spin-assignment
  extension to $C$ (node labels as in Figure~\ref{fig:clause-gadget}):
  (a) when the literal cycle $uwv$ is assigned $+++$ and spins $++-$ 
  are assigned to cycles $v'wu'$ and $v'uw'$, 
  (b) when the literal cycle $v'wu'$ is assigned $+++$ and spins $++-$
  are assigned to cycles $uwv$ and $v'uw'$, and
  (c)  when the literal cycle $v'uw'$ is assigned $+++$ and spins $++-$ 
  are assigned to cycles $uwv$ and $v'wu'$.
Each case is worked out separately in Figures~\ref{fig:onetrue1},
  \ref{fig:onetrue2}, and~\ref{fig:onetrue3}.

\begin{figure}[h]
\centering
\subfigure[]{
\ifpdf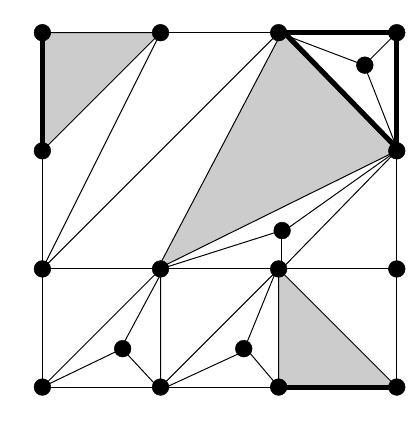\else\input{cl_gad_mix_true_1.pstex_t}\fi
\label{fig:onetrue1}
}
\hspace{2em}
\subfigure[]{
\ifpdf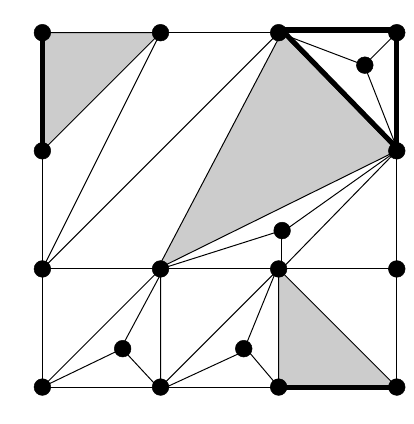\else\input{cl_gad_mix_true_2.pstex_t}\fi
\label{fig:onetrue2}
}\hspace{2em}
\subfigure[]{
\ifpdf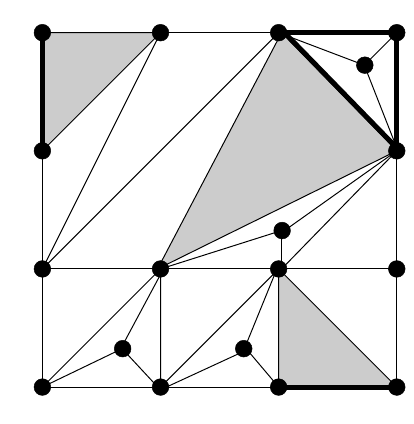\else\input{cl_gad_mix_true_3.pstex_t}\fi
\label{fig:onetrue3}
}
\caption{Unique forced 
  satisfying spin-assignments to a clause gadget 
  when exactly one literal cycle is monochromatic.
Each shown spin-assignment encodes a truth value assignment to 
  the variables of a clause where not all truth values
  are equal.}
\label{fig:assig_exactly_one_true}
\end{figure}

\begin{figure}[ht]
\centering
\subfigure[]{
\ifpdf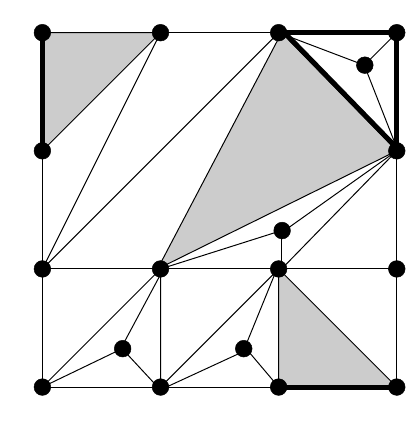\else\input{cl_gad_mix_false_1.pstex_t}\fi
\label{fig:onefalse1}}
\hspace{2em}
\subfigure[]{
\ifpdf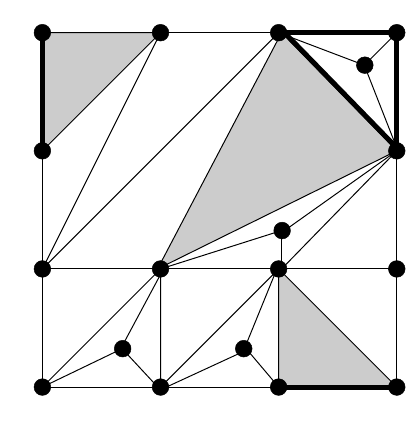\else\input{cl_gad_mix_false_2.pstex_t}\fi
\label{fig:onefalse2}
}\hspace{2em}
\subfigure[]{
\ifpdf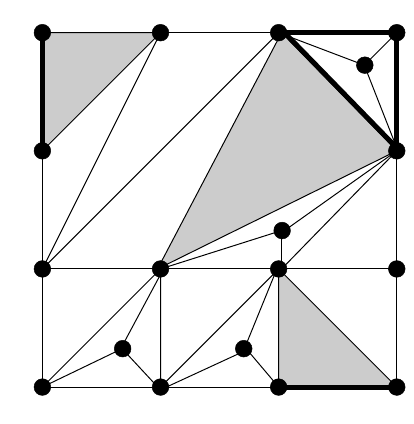\else\input{cl_gad_mix_false_3.pstex_t}\fi
\label{fig:onefalse3}}
\caption{Unique forced 
  satisfying spin-assignments to a clause gadget 
  when exactly one literal cycle is non-monochromatic.
Each shown spin-assignment encodes a truth value assignment to 
  the variables of a clause where not all truth values
  are equal.}
\label{fig:assig_exactly_one_false}
\end{figure}

For the non-monochromatic case, we proceed in the same way. 
Again, By duality and Proposition~\ref{prop:serious_edges_clause},
  it suffices to examine the following situations 
  (node labels as in Figure~\ref{fig:clause-gadget}):
  (a) when the literal cycle $uwv$ is assigned $++-$ and spins $+++$ 
  are assigned to cycles $v'wu'$ and $v'uw'$,
  (b) when the literal cycle $v'wu'$ is assigned $++-$ and spins $+++$
  are assigned to cycles $uwv$ and $v'uw'$, and
  (c)  when the literal cycle $v'uw'$ is assigned $++-$ and spins $+++$ 
  are assigned to cycles $uwv$ and $v'wu'$.
Each case is worked out in Figures~\ref{fig:onefalse1}, \ref{fig:onefalse2},
  and~\ref{fig:onefalse3}.
\end{proof}

\section{The reduction} \label{sec:thereduction}
We now describe the reduction from $\PNAESAT$ to $\SatAssign$.
Let $\varphi$ be a Boolean formula in conjunctive normal form,
  where each clause has exactly three (non-negated) literals. 
Let $x_1,\ldots,x_n$ be 
  the variables and $Cl_1,\ldots,Cl_m$ be the clauses of $\varphi$.
Let $t_i$ denote the number of times variable $x_i$ appears in
  the collection of clauses (multiple occurrences are counted multiple
  times). 
Define 
\begin{displaymath}
k_i = \left\{
\begin{array}{ll}
 2, & \text{if $t_i =1$}, \\
 2 \lceil (1/2)\log_{2}(t_i)\rceil, & \text{if $t_i>1$}.
\end{array} \right.
\end{displaymath}
To each variable $x_i$ we associate a choice gadget $L_i$.
To each clause $Cl_j$ we associate a clause gadget $C_j$.
For $i=1,\ldots,n$, we identify the 
  starting cycle of a $k_i$-replicator
  gadget $R^{k_i}$ with the variable cycle of the choice gadget $L_i$
  so that the fundamental edge of $L_i$ 
  that belongs to $L_i$'s variable cycle and 
  the fundamental edge of $R^{k_i}$ that belongs 
  to its starting cycle coincide. 
Note that the number of 
  end cycles of $R^{k_i}$ is at least $t_i$.
For $i = 1,\ldots,n$, identify $t_i$ 
  end cycles of $R^{k_i}$ with literal cycles 
  of the clause gadgets $C_1,\ldots,C_m$ where variable $x_i$ appears 
  in such a way that the fundamental 
  edges of the end cycles of $R^{k_i}$ and 
  the fundamental edges of the literal cycles of the clause gadgets
  coincide. 
Identify the remaining $2^{k_i}-t_{i}$ end cycles 
  of $R^{k_i}$ (if any) with the outer cycle of a cap gadget.
Denote by $T_{\varphi}$ the surface triangulation thus obtained.

We first make a simple observation.
\begin{lemma}\label{lem:fund_clause_serious}
The fundamental edges of each clause gadget of $T_{\varphi}$ are serious.
\end{lemma}
\begin{proof}
Just observe that by construction of $T_{\varphi}$, 
  fundamental edges of clause gadgets
  are identified with fundamental edges
  of replicator gadgets which are
  known to be serious, as established by  
  Corollary~\ref{coroll:replicatorgadget}~\ref{st1_rep}.
\end{proof}

Note that, by construction, the surface on which~$T_{\varphi}$ is embedded
  is an orientable closed surface of genus
  $m +n + 4 \sum_{i=1}^{n}(2^{k_i}-1)$ ($1$ due to each 
  choice gadget $L_1,\ldots,L_n$, another $1$ due to each
  clause gadget $C_1,\ldots,C_m$, and 
  $4(2^{k_i}-1)$ due to each replicator gadget $R^{k_i}$, $i=1,\ldots,n$).
Clearly, since 
  each $3$-cycle circumscribing gadget holes were identified
  with a $3$-cycle circumscribing another gadget hole, 
  the surface on which~$T_{\varphi}$ is embedded 
  \emph{does not} have holes, i.e.~its a closed surface.
Moreover, since the construction process
  dos not create additional faces,
  each face of $T_{\varphi}$ is a face of some gadget.
Thence, each face is bounded by a $3$-cycle, so $T_{\varphi}$
  is a triangulation.
Finally, note that since each of the gadgets used in the 
  construction of $T_{\varphi}$ is embeddable in an orientable 
  surface with holes, the resulting surface on which $T_{\varphi}$
  is embedded is also an orientable surface.
Summarizing, $T_{\varphi}$ is a triangulation of an orientable
  closed surface.


We say that $\varphi$ is \emph{connected} if for every non-trivial 
  partition $\{S,\bar{S}\}$ of the clauses of $\varphi$ 
  (i.e.~$S,\bar{S}\neq\emptyset$,
  $S\cap\bar{S}=\emptyset$, and $S\cup\bar{S}$ equals the set of clauses 
  of $\varphi$) there
  is at least one variable that appears in one of the clauses 
  in $S$ and in one of the clauses of $\bar{S}$.

We now make a couple of useful observations.
\begin{lemma}\label{lem:connected}
Let $\varphi$ be an instance of $\PNAESAT$.
If $\varphi$ is connected, then the surface in which 
  $T_{\varphi}$ is embedded is also connected.
\end{lemma}
\begin{proof}
Assume $T_{\varphi}$ is embedded in a non-connected surface.
Consider the set $S$ of clauses whose 
  associated clause gadgets are 
  embedded in one of the connected surface components, say $\mathcal{S}$.
Let $\bar{S}$ be the collection of clauses not in $S$.
Note that $\{S,\bar{S}\}$ is non-trivial.
Moreover, the set of variables that appear in clauses in $S$
  (respectively, in $\bar{S}$)
  correspond to those variables associated to choice gadgets 
  embedded in $\mathcal{S}$ (respectively, not in $\mathcal{S}$).
Both of theses collection of variables must be disjoint, contradicting
  the fact that $\varphi$ is connected.
\end{proof}

\begin{lemma}\label{lemm:allsamesign}
Consider an instance $\varphi$ of $\PNAESAT$.
Let  $Cl$ be a clause of $\varphi$ and $x_1$, $x_2$, and $x_3$  
  the (not necessarily distinct) variables appearing in $Cl$.
Let $L_1$, $L_2$, and $L_3$ be the (not necessarily distinct) 
  choice gadgets of $T_{\varphi}$
  associated to $x_1$, $x_2$ and $x_3$, respectively.
Let $C$ be the clause gadget of $T_{\varphi}$ associated to $Cl$.
Then, for every
  satisfying spin-assignment $s$ to $T_{\varphi}$,
  the literal cycles of $C$ and the variable
  cycles of $L_{1}$, $L_{2}$, and $L_{3}$ 
  have the same sign. 
Moreover, if $\varphi$ is connected, then all literal cycles and variable
  cycles of $T_{\varphi}$ have the same sign.
\end{lemma}
\begin{proof}
By Proposition \ref{prop:serious_edges_clause}
  and Lemma~\ref{lem:fund_clause_serious},
  under any satisfying spin-assignment to $T_{\varphi}$
  all literal cycles of $C$ have the same sign. 
Without loss of generality, we can assume that all literal 
  cycles of $C$ are positive.
Hence, the end cycles of the replicator 
  gadgets, say $R^{k_{1}}$, $R^{k_{2}}$, and $R^{k_{3}}$,
  which are identified with the literal 
  cycles of $C$, must all be positive. 
Since $k_1$, $k_2$ and $k_3$ are even, by 
  Corollary~\ref{coroll:replicatorgadget}~\ref{st2_rep}, 
  the starting cycles of $R^{k_{1}}$, 
  $R^{k_{2}}$ and $R^{k_{3}}$ are positive.
Given that the variable cycles of $L_1$, $L_2$, and $L_3$ are 
  identified in $T_{\varphi}$ with the starting cycles of 
  $R^{k_{1}}$, $R^{k_{2}}$ and $R^{k_{3}}$, the first stated claim follows.

The last statement follows trivially from Lemma~\ref{lem:connected}.
\end{proof}

\begin{theorem}\label{theo:main}
Let $\varphi$ be an instance of $\PNAESAT$.
If $\varphi$ is connected, then 
\begin{enumerate}
\item For each truth value assignment that 
  witnesses membership of $\varphi$ in $\PNAESAT$
  there is a unique up to duality satisfying spin-assignment to $T_{\varphi}$.
\item For every pair of duality related
  satisfying spin-assignment to $T_{\varphi}$ there is
  exactly one truth value assignment that witnesses membership
  of $\varphi$ in $\PNAESAT$.
\end{enumerate}  
\end{theorem}
\begin{proof}
Let $a_1,a_2,\ldots,a_n$ be a truth value 
  assignment to the variables $x_1,\ldots,x_n$ that 
  is a witness of membership of $\varphi$ in $\PNAESAT$.
We claim that there is a unique up to duality satisfying
  spin-assignment to $T_{\varphi}$.
As usual, let $Cl_1,\ldots,Cl_m$ be the clauses of $\varphi$,
  let $C_1,\ldots,C_m$ denote the associated clause gadgets,
  and let $L_1,\ldots,L_n$ and $R^{k_1},\ldots,R^{k_n}$ be the 
  choice and replicator gadgets associated to variables 
  $x_1,\ldots,x_n$.

If $a_i$ is \textsc{True}, fix the spins of all nodes 
  of the variable cycle of $L_i$ in such a way that the 
  cycle ends up being positive and monochromatic 
  (observe that this can be done in a unique way).
Otherwise, $a_i$ is \textsc{False}, fix the spins of all nodes 
  of the variable cycle of $L_i$ in such a way that the 
  cycle ends up being positive and non-monochromatic 
  (observe that since fundamental edges of choice gadgets are serious,
  this can again be done in a unique way).
Extend the so far partially defined satisfying spin-assignment 
  to the union of choice gadgets (by Proposition~\ref{prop:choice_properties},
  such a satisfying assignment extension exists and is unique).

Similarly, fix the spins of the nodes of the 
  literal cycles of each clause gadget $C_j$ according to 
  the truth value taken by the associated Boolean formula variable,
  i.e.~if the variable's value is \textsc{True}, make 
  the literal cycle positive and monochromatic, and 
  positive and non-monochromatic otherwise
  (observe again that such spin-assignments 
  can be done in a unique way).
Extend once more the so far partially defined spin-assignment 
  to the union of clause gadgets.
We claim that such an extension exists and is unique.
Indeed, given that $a_1,\ldots,a_n$ is a witness of membership of 
  $\varphi$ in $\PNAESAT$, the variables in each clause $Cl_j$ 
  do not take the same truth value under the assignment $a_1,\ldots,a_n$.
Thus, the aforementioned spin-assignment to the literal cycles of 
  $C_j$ is such that not all literal cycles 
  end up having the same chromaticity.
By Proposition~\ref{prop:clause-feasible}, 
  each clause gadget has a unique satisfying spin assignment extension,
  thus establishing our claim.

Recall that variable cycles of choice gadgets 
  (respectively, literal cycles of clause gadgets) are identified with 
  starting cycles (respectively, end cycles) of replicator gadgets.
Hence, for all $i=1,\ldots,n$, 
  the partial spin-assignment thus far defined makes 
  the start and end cycle of the replicator gadget $R^{k_i}$ 
  positive and monochromatic (respectively, non-monochromatic)
  if and only if $a_i$ is \textsc{True} (respectively, \textsc{False}).
Given that $k_i$ is even and positive for all $i = 1,\ldots, n$,
  by Corollary~\ref{coroll:replicatorgadget}~\ref{st4_rep}
  and Proposition~\ref{prop:cap}, there is a unique extension 
  of the previously defined partial spin-assignment to all nodes of 
  replicator gadgets so the resulting spin-assignment is a satisfying
  spin-assignment for $T_{\varphi}$.

We now prove the second part of the claimed result.
Assume there is a satisfying spin-assignment~$s$ to $T_{\varphi}$. 
By Lemma~\ref{lemm:allsamesign} and since $\varphi$ is connected, 
  all literal and variable cycles of $T_{\varphi}$ have the same
  sign, say  positive.
For $i=1,\ldots,n$, let $a_i$ be \textsc{True} if the variable cycle 
  of $L_i$ is positive and monochromatic, 
  and \textsc{False} if the variable
  cycle of $L_i$ is positive and non-monochromatic. 
We claim that $a_1,\ldots,a_n$ is a witness of membership of 
  $\varphi$ in $\PNAESAT$.
Indeed, assume $x_{j_1}$, $x_{j_2}$, and $x_{j_3}$ are the
  (not necessarily distinct) variables appearing in clause $Cl_j$.
Let $s\in\{1,2,3\}$.
Since the start cycle of the replicator gadgets $R^{k_{j_s}}$
  is identified with the variable
  cycle of the choice gadget $L_{j_s}$, 
  then by Corollary~\ref{coroll:replicatorgadget} the 
  end cycles of $R^{k_{j_s}}$ must be monochromatic if and only if
  $a_{j_s}$ is \textsc{True}.
Since the non-capped end cycles of replicator gadgets are identified
  with the literal cycles of clause gadgets, we have that 
  the literal cycle of the clause gadget $C_j$ associated to 
  the variable $x_{j_s}$ is monochromatic if and only if 
  $a_{j_s}$ is \textsc{True}.
Moreover, all literal cycles of clause gadgets are positive.
Since $s$ is a satisfying spin-assignment, 
  by Proposition~\ref{prop:clause-nonfeasible} and 
  Proposition~\ref{prop:clause-feasible},
  the literal cycles of $C_j$ can not all be either 
  monochromatic or non-monochromatic.
This implies that $a_{j_1}$, $a_{j_2}$, and $a_{j_3}$ are not
  all equal, as we wanted to establish.
\end{proof}

\begin{corollary}\label{coroll:NPhard} 
Let $\varphi$ be an instance of $\PNAESAT$.
Then, $\varphi$ is satisfiable if and only if there is a satisfying
  spin-assignment to $T_{\varphi}$.
\end{corollary}
\begin{proof}
If $\varphi$ is connected, the result is immediate from
  Theorem~\ref{theo:main}.
Assume $\varphi$ is not connected.
Then, there are $\varphi_1,\ldots,\varphi_c$ instances of 
  $\PNAESAT$ such that each $\varphi_i$ is connected and 
  $\varphi=\wedge_{i=1}^{c}\varphi_i$.
Moreover, $\varphi$ belongs to $\PNAESAT$ if and only if
  $\varphi_1,\ldots,\varphi_c$ belong to $\PNAESAT$.
By Theorem~\ref{theo:main}, this is equivalent to saying 
  that $\varphi$ belongs to $\PNAESAT$ if and only 
  the union of the surface triangulations $T_{\varphi_1},
  \ldots, T_{\varphi_c}$, i.e.~$T_{\varphi}$, admits a satisfying spin
  assignment.
\end{proof}

The next result allows us to handle, in our reduction, instances 
  of $\PNAESAT$ which are not connected.
\begin{lemma}
Let $\varphi$ be an instance of $\PNAESAT$.
Then, there exists a log-space (hence, polynomial time)
  computable instance $\varphi'$ of $\PNAESAT$ such that
  $\varphi'$ is connected and the following equality holds:
\begin{eqnarray*}
\lefteqn{\left|\{\vec{a}=(a_1,\ldots,a_n): 
  \text{$\vec{a}$ is a witness of membership of $\varphi$
  in $\PNAESAT$}\}\right|
  =  } \\
  & & \quad 
  \frac{1}{2}\left|\{\vec{a'}=(a'_1,\ldots,a'_{n'}): 
  \text{$\vec{a'}$ is a witness of membership of $\varphi'$
  in $\PNAESAT$}\}\right|.
\end{eqnarray*}
\end{lemma}
\begin{proof}
Assume $x_1,\ldots,x_n$ are the variables and $C_1,\ldots,C_m$
  the clauses of $\varphi$. 
Consider two additional Boolean variables $y$ and $z$
  and define $n$ additional clauses $C'_1,\ldots,C'_n$ such that
  $C'_i=x_i\wedge y\wedge z$.
Let $\varphi'$ be the conjunction of $C_1,\ldots,C_n,C'_1,\ldots,C'_n$.
Clearly, $\varphi'$ is a connected instance of $\PNAESAT$ which is 
  log-space computable given $\varphi$.
Note that a membership witness $a_1,\ldots,a_s$ of an instance of 
  $\PNAESAT$ can not be such that all $a_i$'s are equal. 
This immediately implies that
  $x_1,\ldots,x_n$ is a witness of membership of $\varphi$
  in $\PNAESAT$ if and only if 
  $x_1,\ldots,x_n,y=0,z=1$ and 
  $x_1,\ldots,x_n,y=1,z=0$ are witnesses of membership of $\varphi'$  
  in $\PNAESAT$.
\end{proof}

\begin{corollary}\label{coroll:sharpPhard} 
Let $\varphi$ be an instance of $\PNAESAT$.
Then, there exists a log-space (hence, polynomial time)
  computable instance $\varphi'$ of $\PNAESAT$ such that
\begin{eqnarray*}
\lefteqn{\left|\{s: \text{$s$ is a satisfying spin-assignments to 
  $T_{\varphi'}$}\}\right|
  =  } \\
  & & \quad 
  4\cdot\left|\{\vec{a}=(a_1,\ldots,a_{n}): 
  \text{$\vec{a}$ is a witness of membership of $\varphi$
  in $\PNAESAT$}\}\right|.
\end{eqnarray*}
\end{corollary}

To conclude, note that given an instance $\varphi$ of 
  $\PNAESAT$ with $n$ variables and $m$ clauses each of the 
  necessary gadgets can be constructed
  in $O(\log|V(T_\varphi')|)$ space, i.e.~in log-space in the size
  of the encoding of $\varphi$.
The number of choice, block-replicator, and clause gadgets that need to 
  be built are $n+2$, 
  $\sum_{i=1}^{n}(2^{k_i} -1) = O(m)$, and $m+n$, respectively. 
Hence, an encoding of a rotation system for $T_{\varphi'}$ can 
  be computed in log-space, thence 
  in polynomial time.
Then, Corollary~\ref{coroll:NPhard} and Corollary~\ref{coroll:sharpPhard}
  imply Theorem~\ref{th:sat} and Theorem~\ref{th:count}, respectively.

\medskip
To conclude, note that this work does not preclude that the problem 
  of deciding whether or not 
  a triangulation embedded in surfaces of a fixed constant genus
  admits a satisfying spin-assignment might be say decidable in polynomial
  time.
Neither does it preclude that the associated counting problem might
  be significantly easier than $\sharpP$-hard.
We believe this are interesting issues to study.

\section*{Acknowledgements}
The authors thank Martin Loebl for his enthusiasm, motivation, and 
  many helpful discussions.

\end{document}